\documentclass[12pt,a4paper]{article}

\usepackage[utf8]{inputenc}
\usepackage[english]{babel}
\usepackage{graphicx}
\usepackage{multirow}
\usepackage{caption}
\usepackage{color}
\usepackage{cancel}
\usepackage{bm}
\usepackage{url}

\usepackage{subfigure}

\usepackage{amssymb,amsmath,amsthm}
\usepackage{vmargin}

\theoremstyle{plain}
\newtheorem{theorem}{Theorem}

\newtheorem{corollary}[theorem]{Corollary}
\newtheorem{proposition}[theorem]{Proposition}

\theoremstyle{definition}

\theoremstyle{remark}
\newtheorem{remark}[theorem]{Remark}

\numberwithin{equation}{section}
\numberwithin{theorem}{section}

\newcommand{\G}{\mathbb G}

\newcommand{\1}{\mathbf 1}
\newcommand{\g}{\mathbf{g}}

\setpapersize{A4}
\setmargins{2.5cm}{2cm}{16cm}{25cm}{0cm}{0cm}{1cm}{1cm}

\usepackage{sectsty}
\sectionfont{\large}

\begin{document}

\title{A test of hypotheses for random graph distributions built from EEG
data}

\author{Andressa Cerqueira, Daniel Fraiman, \\Claudia D. Vargas and Florencia Leonardi}

\date{\today}

\maketitle

\begin{abstract}
The theory of random graphs is being applied in recent years to model neural interactions in the brain. While the probabilistic properties of random graphs has been extensively studied in the literature, the development of statistical inference methods for this class of objects
has received less attention.
In this  work we propose a non-parametric test of hypotheses to test
if two samples of random graphs were originated from the same probability distribution. We show
how to compute efficiently the test statistic and we study its
performance on simulated
data.
We apply the test to compare graphs of brain functional network interactions built from eletroencephalographic (EEG) data collected during the visualization of point light displays depicting human locomotion.

\end{abstract}

\section{Introduction}

The brain consists in a complex network of interconnected regions whose functional interplay is thought to play a major role in cognitive processes \cite{aertsen,friston93,vander10}. Based on an elegant representation of nodes (vertices) and links (edges) between pairs of nodes where nodes usually represent anatomically defined brain regions while links represent functional or effective connectivity \cite{bullmore2009}, random graph theory is progressively allowing to explore properties of this sophisticated network  \cite{calmels,fraiman}.  Such  properties have been used so far to infer, for instance, about effects of brain lesion \cite{wang}, ageing \cite{achard,wu,meunier} and neuropsychiatric diseases (for a recent review, see \cite{bassett09}).

From a theoretical point of view, the most famous model of random graphs is the Erdös-Reny  model  \cite{erdos-renyi-1960} (introduced by Gilbert in \cite{gilbert1959}), where the edges of the graphs  are independent and identically distributed Bernoulli random variables.  Besides its simplicity, this model continues to be actively studied  and new properties are being discovered (see for example \cite{chatterjee2011} and references therein).  From the applied point of view the most popular model is the Exponential Random Graph Model (ERGM) that has emerged mainly in the Social Sciences community (see \cite{rinaldo2009} and references therein).

Notwithstanding the crescent interest of the scientific community in the graph theory applications, the development of statistical techniques to compare sets of graphs or network data is still quite limited. Some recent works have addressed the problem of maximum likelihood estimation in ERGM (\cite{rinaldo2009,snijders2010,chatterjee2013}),
 but the testing problem has been even less developed.
  As far as we know, the testing problem is restricted only to the identification of differences in some one dimensional graph property (\cite{fraiman,stam,ibanez}).
At this point it is important to remark that the number of different graphs with $v$ nodes
grows as fast as  $2^{v(v-1)/2}$  which in practice is far  much larger than a typical  sample size analyzed. This is the reason why the testing problem is difficult and relevant given that the graph space has no total order.

 In this paper we propose a goodness-of-fit  test of hypothesis for random graph distributions. The statistic is inspired in a recent work \cite{busch2009} where a test of hypothesis for random trees is developed. We show how to compute the test statistic efficiently  and we prove a Central Limit Theorem.
The test makes no assumption on the specific form of the distributions and it is consistent against any alternative hypothesis that differs form the sample distribution on at least one edge-marginal.  In a simulation study we show  the efficiency of the test and we compare its performance with the simultaneous testing of  the edge-marginals.
 We also apply the test to compare graphs built from  electroencephalographic (EEG) signals collected during the observation of videos depicting human locomotion.

\section{Definition of the test}

Let $V$ denote a finite set of vertices, with cardinal $|V|=v$, and let $\G(V)$ denote the set of all simple undirected graphs over $V$. We identify a graph $g=(V,E)$ with the indicator function $g_{ij}=\1\{(i,j)\in E\}$.
Given a graph $g\in\G(V)$, we denote by $\1-g$ the graph
defined by $(\1-g)_{ij} = 1-g_{ij}$.

 In order to measure a discrepancy between two graphs $g,g'\in\G(V)$ we introduce a distance given by
 \[
 D(g,g') = \sum_{ij}    (g_{ij} - g'_{ij})^2 \,.
  \]
Here and throughout the rest of the paper summations will refer to the set of vertices $(i,j)\in V^2$ such that $i<j$ (because $g_{ij}=g_{ji}$).

Given a set of graphs $\mathbf{g} = (g^1,\dotsc, g^n)$ and a graph $g\in \G(V)$, we denote by
$\bar D_{\mathbf{g}}(g)$ the mean distance of graph $g$ to the set
$\mathbf{g}$; that is
\[
\bar D_{\mathbf{g}}(g) = \frac1n \sum_{k=1}^n D(g,g^k)\,.
\]
We also define the function $\overline{\g}\colon V^2\to[0,1]$, the \emph{mean} of $\g$, by
\[
\overline\g_{ij}\; =\; \frac1n\sum_{k=1}^n g^k_{ij}\,.
\]

Assume  $g$ is a random graph with  distribution $\pi$. Denote by
$\pi_{ij}=\pi(g_{ij}=1)$ and let $\Sigma$ denote the covariance matrix of
 $\pi$.
Given another probability distribution $\pi'$ defined on $\G(V)$, we are interested in testing the hypothesis
\begin{equation}\label{hip}
H_0\colon \pi=\pi'\quad\text{ versus }\quad H_A\colon \pi\neq\pi'\,.
\end{equation}
Given an i.i.d sample of  graphs $\mathbf{g} = (g^1,\dotsc, g^n)$ with distribution $\pi$, we define the one-sample test statistic
\begin{equation}\label{onesample}
W(\mathbf{g}) = \max_{g\in \G(V)} \,| \bar D_{\mathbf{g}}(g) -  \pi' D(g,\cdot)\,|\,,
\end{equation}
where $\pi'D(g,\cdot)$ denotes the mean distance of graph $g$ to a random graph with distribution $\pi'$ and is given by
\[
\pi'D(g,\cdot) = \sum_{g'\in\G(V)} D(g,g')\pi'(g')\,.
\]

In the same way, given two samples $\mathbf{g} = (g^1,\dotsc, g^n)$ and  $\mathbf{g'} = (g'^1,\dotsc, g'^m)$
with distributions $\pi$ and $\pi'$ respectively,
we define the two-sample test statistic
\begin{equation}\label{twosample}
W(\mathbf{g},\mathbf{g'}) = \max_{g\in \G(V)} \,| \bar D_{\mathbf{g}}(g) - \bar D_{\mathbf{g'}}(g)\,|\,.
\end{equation}

At first sight the computation of (\ref{onesample}) or (\ref{twosample}) is prohibited for even a small number of vertices. But
as we show in the following proposition, it is possible
 to compute the test statistic in $O(v^2(n+m))$ time.

\begin{proposition}\label{main-teo}
For the one-sample test statistic we have that
\begin{equation}\label{wg}
W(\mathbf{g})  \;=\; \sum_{ij} \;\bigl|\, \overline\g_{ij}- \pi'_{ij} \,\bigr|\,.
\end{equation}
Analogously, for the two-sample test statistic we have that
\begin{equation}\label{wgg}
W(\mathbf{g},\mathbf{g'})  \;=\;  \sum_{ij} \;\bigl|\, \overline\g_{ij}-
\overline\g'_{ij}  \,\bigr|\,.
\end{equation}
\end{proposition}

As a corollary of this proposition we prove the following result about the asymptotic distribution of the test statistic. Let $\Pi=(\pi_{ij})$,  $\hat\Pi=(\overline\g_{ij})$ and
$\hat\Pi'=(\overline\g'_{ij})$. Then we can write
$W(\mathbf{g})=\|\hat\Pi -\Pi\|$ and $W(\mathbf{g},\mathbf{g'})=\|\hat\Pi -\hat\Pi'\|$, where $\|\cdot\|$ denotes the
vectorized 1-norm.

\begin{corollary}\label{main-cor}
Under $H_0$, for the one-sample test statistic we have that
\begin{equation*}
\sqrt{n}\bigl(\hat\Pi-\Pi\bigr) \xrightarrow[n \rightarrow \infty]{D} N(0,\Sigma)\,.
\end{equation*}
Analogously, for the two-sample test statistic we have that
\begin{equation*}
\sqrt{\frac{nm}{n+m}}\bigl(\hat\Pi-\hat\Pi'\bigr) \xrightarrow[n,m \rightarrow \infty]{D} N(0,\Sigma)\,.
\end{equation*}
\end{corollary}

The proofs of these results are  postponed to the Appendix.
 
Assuming the distribution of $W$ under the null  hypothesis is  known,
the result of the test \eqref{hip} at the significance level $\alpha$
is
\[
\text{Reject }H_0\quad\text{ if }\quad W(\mathbf{g}) > q_{1-\alpha}\,,
\]
where $q_{1-\alpha}$ is the $(1-\alpha)$-quantile of the distribution of $W$
under the null hypothesis. The result of the test for the two-sample case is obtained in the same way replacing  $W(\mathbf{g})$ by $W(\mathbf{g},\mathbf{g'})$.

\begin{remark}By the form of the resulting test statistic, given in \eqref{wg} and \eqref{wgg},  and Corollary~\ref{main-cor} we can deduce that the test is consistent against any alternative hypothesis $\pi'$ with $\pi'_{ij}\neq\pi_{ij}$ for at least one pair $ij$.
\end{remark}

\section{Performance of the test on simulated data}

In this section we present the results of a simulation study in order to  evaluate the performance of the test \eqref{hip}.
In the first simulation example we compute the power function of the one-sample test of a (modified) Erdös-Renyi model of parameter $p\in(0,1)$ with $v=10$ vertices, taking as null hypothesis the classical Erdös-Renyi model with $p_0=0.5$. In the modified model, a percentage $q$ of the edges of the graph (previously chosen) are independent Bernoulli variables with parameter $p$, and the remaining edges are taken with parameter $p_0$ as in the null model. The power function of the test  \eqref{hip}  as a function of $p$ and for different values of $q$
is presented in Fig.~\ref{fig:1a}. The sample size was $n=20$ and the quantile of the distribution of $W$ (under $H_0$) was computed as the empirical $0.95$ quantile of a simulation with $p_0=0.5$ and $10.000$ replications.
Even for a somehow small  proportion of $25\%$ of different edge probabilities and a small sample size the test performs well and the power function approaches 1 when the absolute difference $|p-p_0|$ grows.
In order to compare our  results with a classical method, we performed simultaneous hypothesis tests on the edge occurrences by using Bonferroni correction (BC). In this case exact critical regions were obtained from the Binomial distribution. For all values of $q$, the $W$ test performs better than the BC procedure, as shown in  Fig.~\ref{fig:1a}.

\begin{figure}[t!]
\center
\vspace*{-0.5cm}
\subfigure[centerlast][W test]{\includegraphics[scale=0.74]{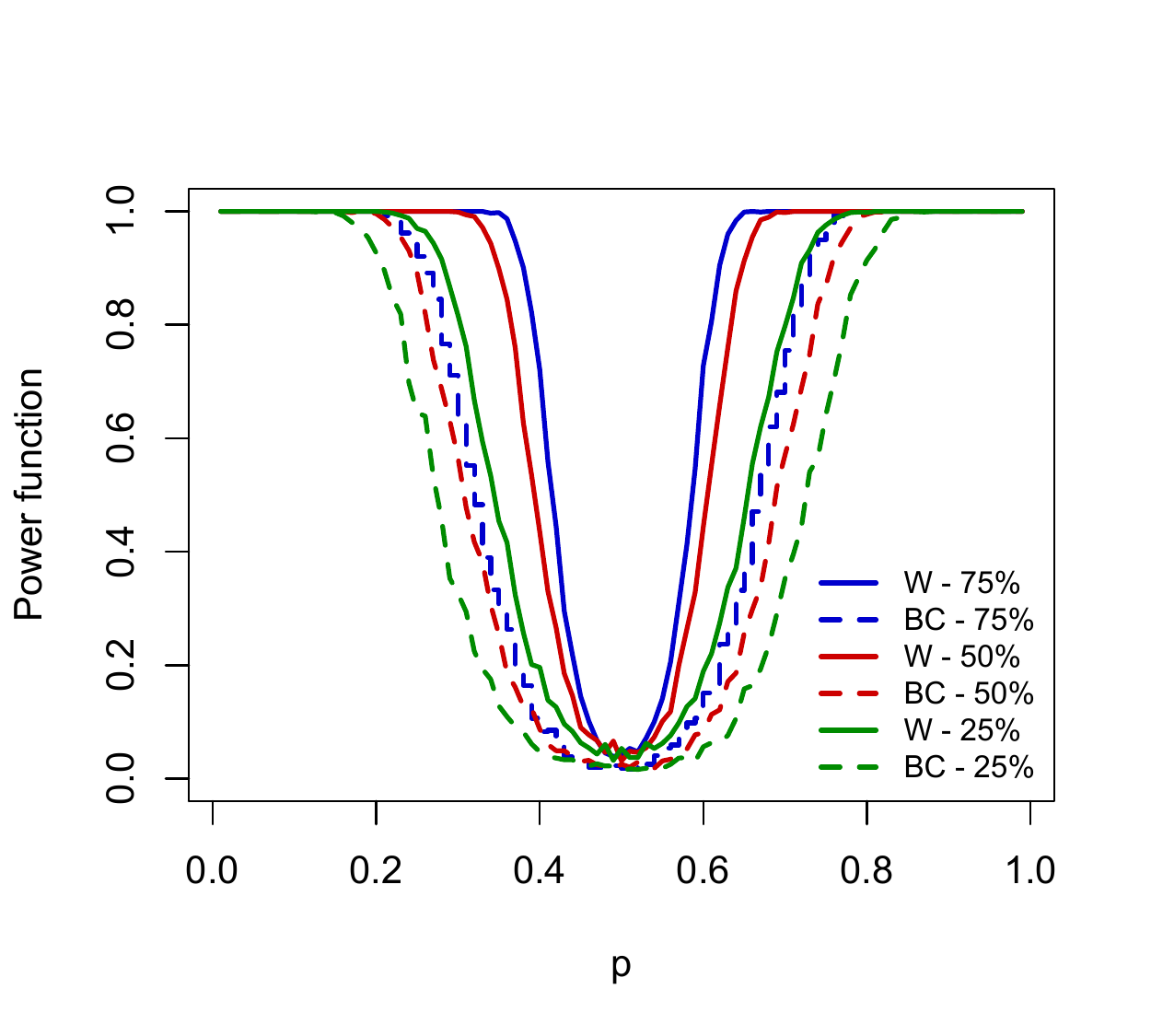}
\label{fig:1a}}
\subfigure[centerlast][W test]{\includegraphics[scale=0.74]{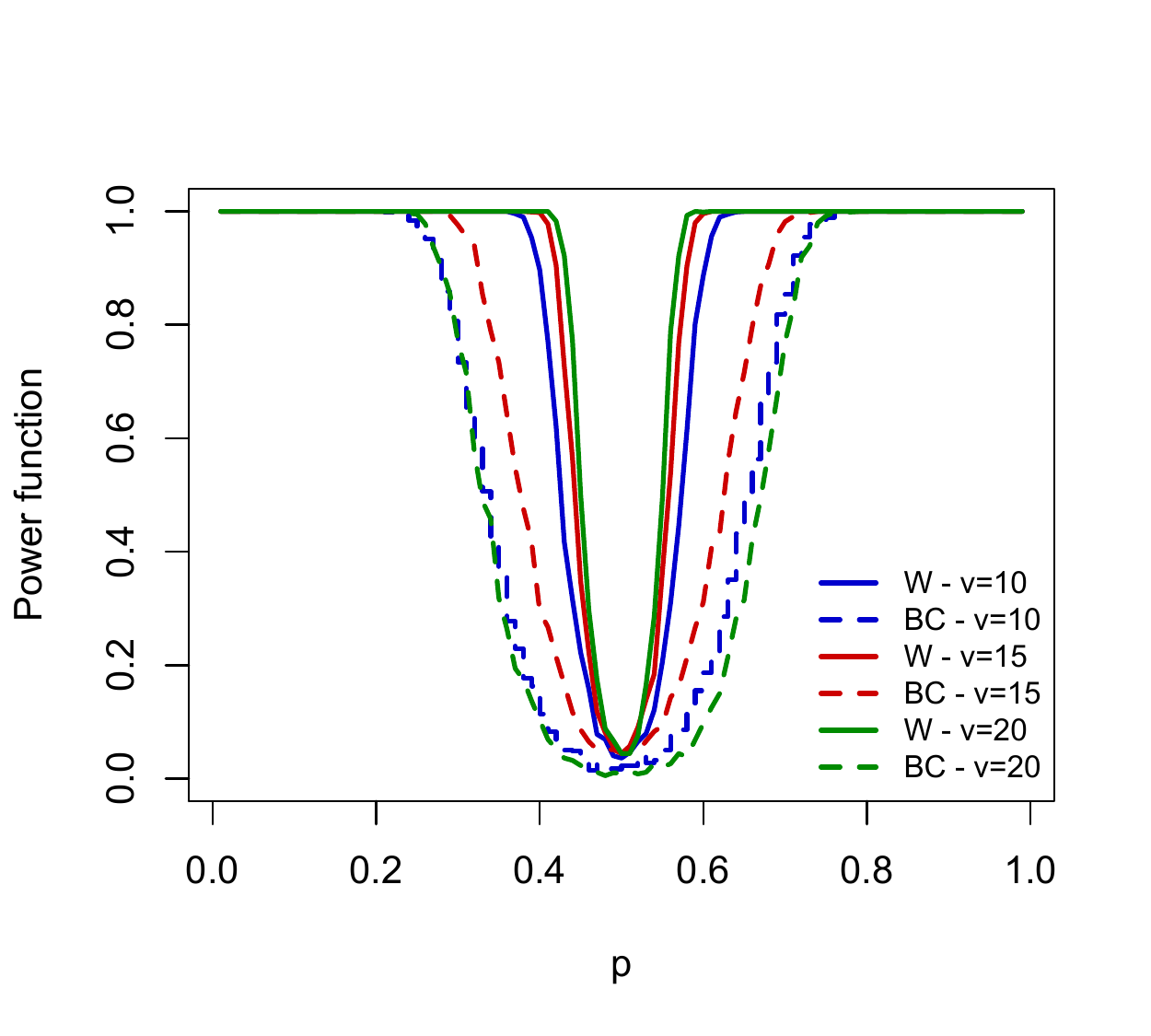}
\label{fig:1b}}
\caption{Comparison of the power functions of the one-sample $W$ test and the simultaneous testing procedure with Bonferroni correction (BC). The null model is Erdös-Renyi of parameter $p_0=0.5$ and the alternative hypothesis is (a) (modified) Erdös-Renyi  model  with $v=10$ nodes and $q\%$  of edges with parameter $p$ and the remaining edges with parameter $p_0=0.5$ and (b) Erdös-Renyi  model  with parameter $p$ and a different number $v$ of nodes in the graphs. }
\label{fig:erdosfracs}
\end{figure}

Alternatively, in Fig.~\ref{fig:1b} we show a comparison of the power function of a pure Erdös-Renyi model ($q=100\%$) but with different graph sizes (given by the number of nodes $v$).
As in the first example  we took as null model the  Erdös-Renyi model of parameter $p_0=0.5$.
 In all the simulations the sample size was $n=20$.   In this case, as it can be expected the power function is closer to 1 as the number of vertices increases because there is  more evidence against the null hypothesis. We emphasize the good performance of the $W$ statistic taking into account the small sample size and the number of possible graphs (that in the case of $v=10$ is $|\G(V)| = 2^{55}$ graphs).
In the same figure we show the comparison with the power functions for the BC procedure.
For all values of $v$ considered the $W$ statistic outperforms the simultaneous tests with Bonferroni correction.   Moreover, the test proposed here gains power as the number of vertices grows meanwhile the BC procedure decreases its power.

Finally, we focus our attention in the power function for the Exponential Random Graph
Model (ERGM).
In this model, the probability of a graph  $g$ is given by
\begin{equation}\label{prob}
\pi(g|\theta)=\frac{\exp(\theta\cdot S(g))}{z(\theta)}\,,
\end{equation}
where $\theta=(\theta_1,\dots,\theta_k)$ is the parameter vector,
$S(g)$ is a vector of $k$ statistics computed from $g$ (e.g. the number of edges, degree statistics, triangles, etc.) and $z(\theta)$ is the normalizing constant. Depending on the particular $S(g)$ function and parameter vector  $\theta$, the model favors graphs with distinct small structures. For example, the so called edge-triangle model, where $S(g)= (n_e(g),n_t(g))$ is the number of edges and the number of triangles present in $g$, penalizes (when $\theta_2<0$) or favors (when $\theta_2>0$) the appearance of triangle structures on the graph. The statistic $S(g)$ can be simply computed from $g$ by the formulas  $n_e(g)=\sum_{ij}g_{ij}$ and $n_t(g)=\sum_{ijk}g_{ij}g_{jk}g_{ik}$. In the model with $\theta_1\geq0$  and $\theta_2>0$ ($\theta_1\leq0$, $\theta_2<0$) it happens that the graph with highest probability (eq.~\ref{prob}) is the complete (null) graph where all edges are present (absent).
When $\theta_1<0$ and $\theta_2>0$, the model favors the appearance of triangles in the graphs but penalizing graphs with too many edges.
This model is very sensitive to the values of the parameter vector and is not uniquely determined, meaning that different values of the parameter vector can give rise to the same probability distribution. Even more, for $\theta_1\in \mathbb R$ and $\theta_2>0$, when the number of nodes increases, the ERGM model  is closed  in distribution to  a Erdös-Renyi model  (cf. \cite{chatterjee2013} and references therein).
Another well studied model is the edge-2star model, defined by $S(g)=(n_e(g),n_s(g))$, with $n_s(g)=\sum_{ijk}g_{ij}g_{jk}$, which generates small graphs with nodes of degree 2 or more (less) for $\theta_2>0$ ($\theta_2<0$). In this last model there is no ``incentive'' for the two nodes that join one of degree 2 to be linked creating a triangle structure.

\begin{figure}[t!]
\center
\subfigure[centerlast][Edge-triange model]{\includegraphics[scale=0.75]{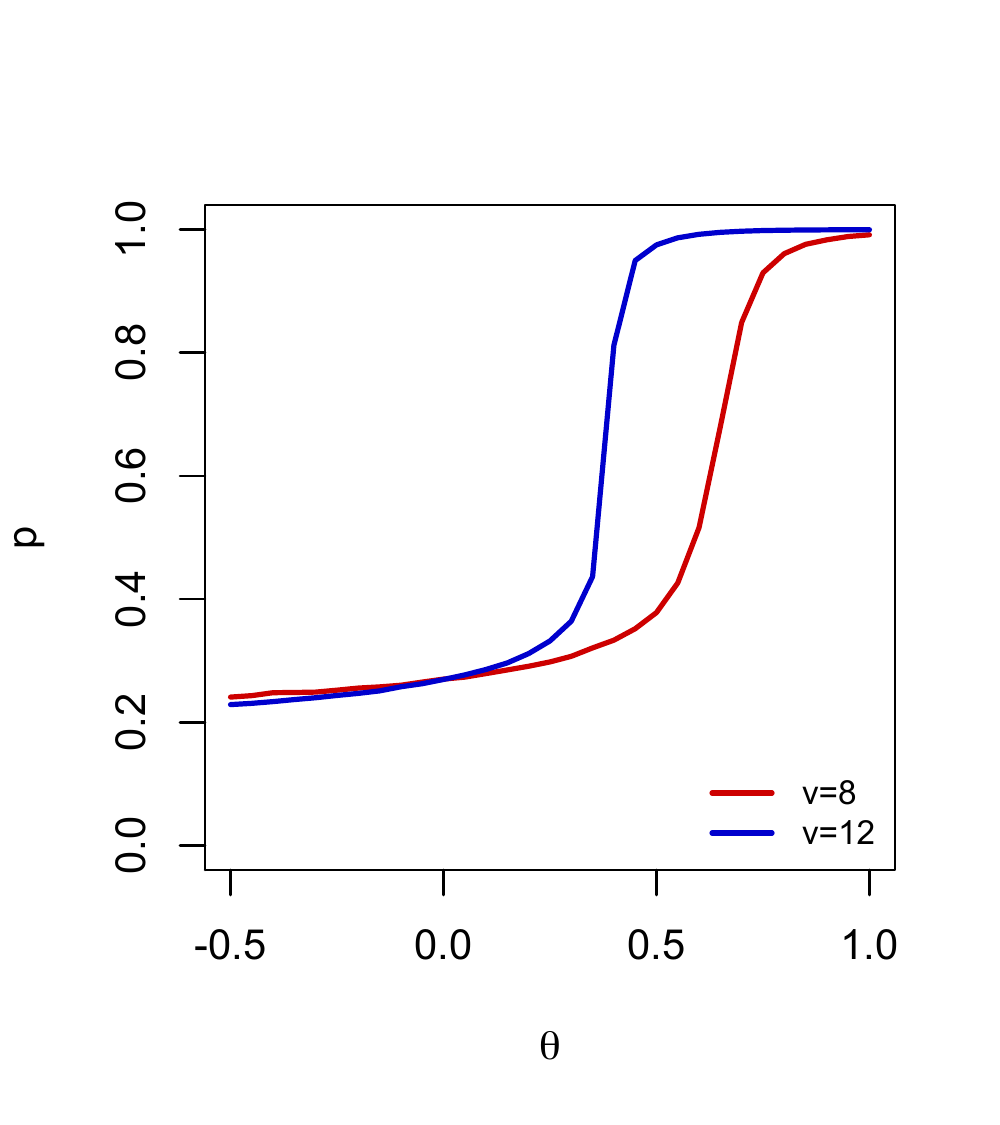}
\includegraphics[scale=0.72]{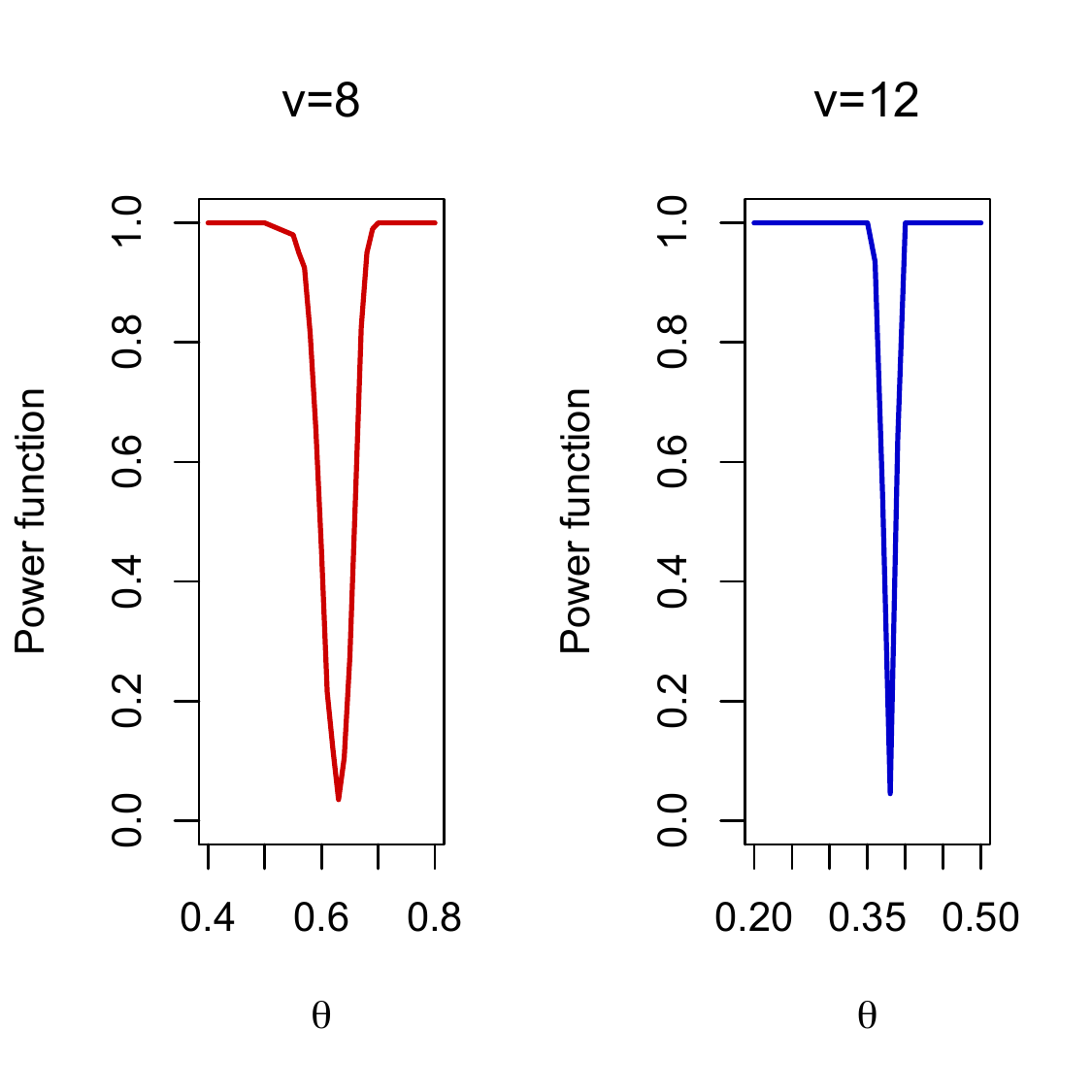}
\label{fig:theta-p}}
\subfigure[centerlast][Edge-2star model]{\includegraphics[scale=0.75]{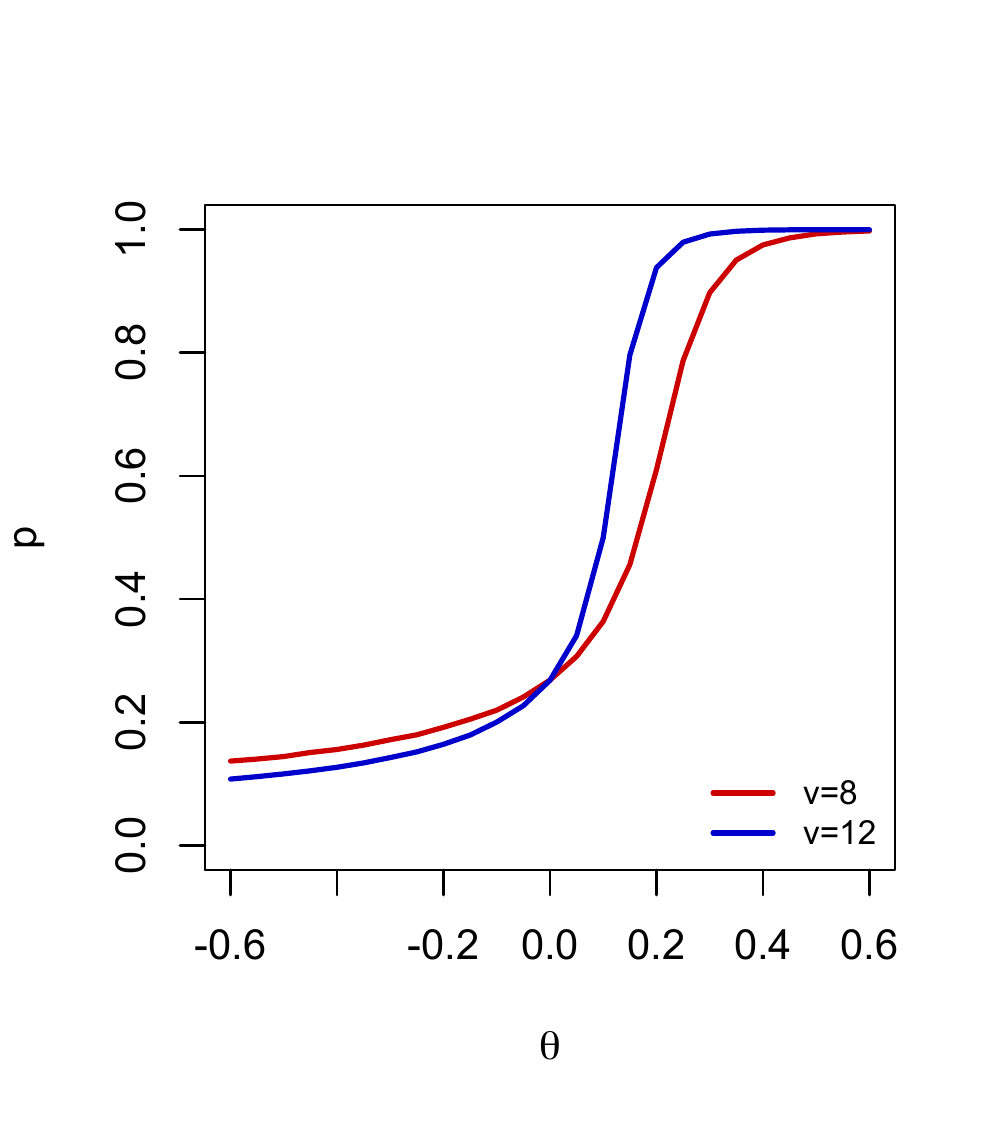}
\includegraphics[scale=0.72]{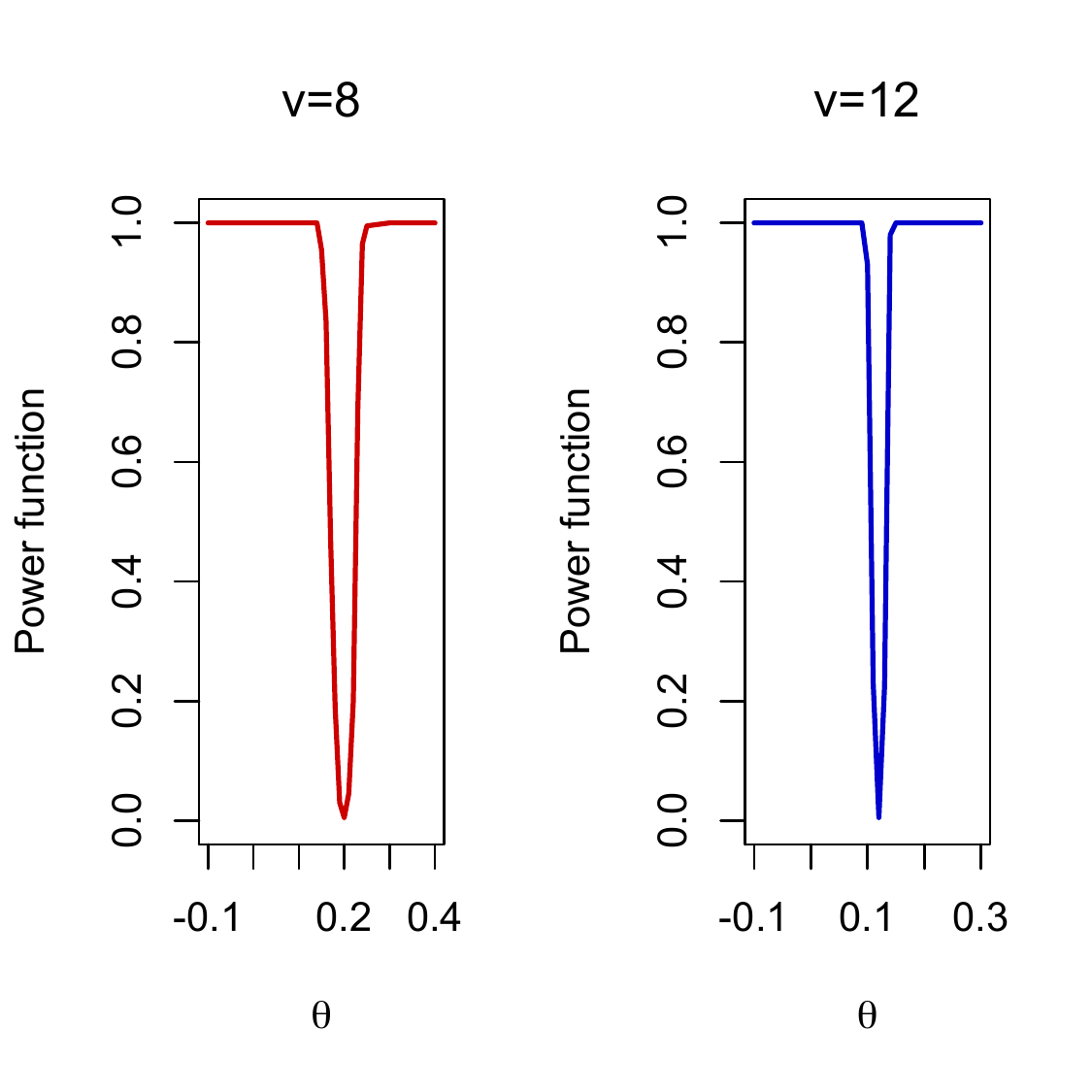}
\label{fig:theta-p-kstar}}
\caption{Power of the two sample test for comparing the edge-triangle model with $(\theta_1,\theta_2)=(-0.1,0.1)$ against (a)
the edge-triangle model with $(\theta_1,\theta_2)=(-0.1,\theta)$ and (b)
the edge-2-star model with parameter $(\theta_1,\theta_2)=(-0.1,\theta)$. }
\label{fig:ergm}
\end{figure}

In our first evaluation of the $W$ statistic for the ERGM  we consider  the edge-triangle model with parameter vector $(\theta_1,\theta_2)=(-1,\theta)$ and two different values of $v$. To understand the behavior of this model as a function of the parameter vector, we first compute the density of edges  $p$  for each value of $\theta$ ranging from $-0.5$ to $1$, the results are summarized in Fig.~\ref{fig:theta-p}. We can observe that the density of edges grows very fast in the interval $(0.5,0.8)$ for $v=8$ and $(0.3,0.5)$ for $v=12$ and as we will see, this fact is relevant for the  behavior of the power function.
To consider moderate null models (far from being the full graph or the null graph) we take for each value of $v$ the corresponding value of $\theta$ given a density of edges approximately equal to $p=0.6$, this corresponds to $\theta_2=0.63$ for $v=8$ and $\theta_2=0.38$ for $v=12$. For each one of these null models,  we computed the
$(1-\alpha)$-quantile of the distribution of $W$ under $H_0$ for the one sample test statistic \eqref{wg}, using 2000 replications of  $W$  with sample size $n=20$.
Then we computed the power function of the test against any hypothesis with $\theta_2=\theta$, with $\theta$ ranging from $-0.5$ to $1$.
We can see that the power functions grow very fast to 1 and this is a consequence of $\theta_2$ in the null model  being in the interval where $p$ grows very fast. The behavior of the power function is very anomalous if we take as null model a value of $\theta_2$ belonging to a flat region of $p$, because in these cases a big difference in $\theta$ does not imply a big difference in $p$, and this is determinant for the value of the power function.

In our second example we took the same null models for each one of the two cases $v=8$ and $v=12$ and computed the power function
for the  edge-2star  model with parameter vector $(\theta_1,\theta_2)=(-1,\theta)$, with
$\theta$ varying between $-0.6$ and $0.6$. As in our previous example, we computed the density of edges $p$ for each value of $\theta$ (Fig.~\ref{fig:theta-p-kstar}). In these cases the power function also converges very fast to 1 because the value $p=0.6$ of the null hypothesis lies in a region where the density of edges in the edge-2star model also grows very fast as a function of $\theta$.

\section{Discrimination of EEG brain networks}

\begin{figure}[t!]
\centerline{\includegraphics[scale=1.05]{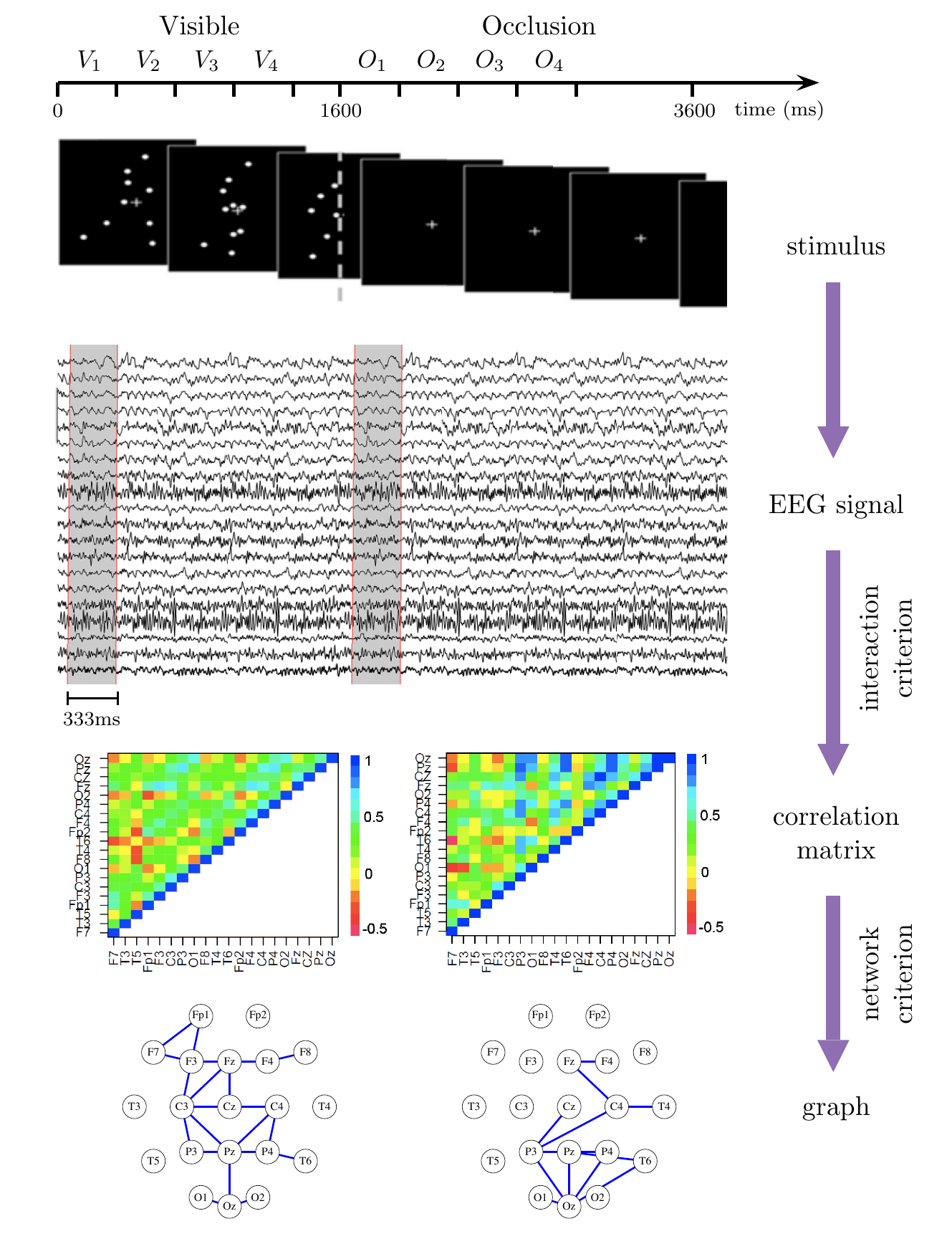}}
\caption{First two phases of the stimulus used in the experiment and the steps to obtain the samples of graphs from the EEG signal.
}
\label{fig:estimulo}
\end{figure}

The data analyzed in this section where first presented in \cite{ghislain}.
A total of sixteen healthy subjects (29.25$\pm$6.3 years) with normal or corrected to normal vision and with no known neurological abnormalities participated in this study. The study was conducted in accordance with the declaration of Helsinki (1964) and approved by the local ethics committee (Comit\'e de \'Etica em pesquisa do Hospital Universit\'ario Clementino Fraga Filho, Universidade Federal do Rio de Janeiro, 303.416).

The stimulus used in the experiment is composed by 10 white luminous points with black background, that represent 10 markers of the human body (head, shoulder, elbow, hand, hip, knee and ankle). The animation of these points permitted a vivid perception of a walker's movement (which we call biological movement).
The stimulus has a total length of 5200ms and is composed by 3 different phases: the visible phase (0 - 1600ms) represents the individual walking, the occlusion phase (1600ms - 3900ms) where the luminous points disappear behind a black wall and the phase of reappearance (3900ms - 5200ms) where the individual is again visible and continues walking.
A second stimulus employed in that study consisted on a permuted version of the point lights, thus destroying the gestalt of the human walker motion. This stimulus is called non-biological movement.
The results presented in this paper only consider the visible and the occlusion phases of the experiment (0 - 3900ms), a representation of the stimuli can be observed in the top of  Fig.~\ref{fig:estimulo}.

The EEG activity  was registered using a BrainNet BNT 36 (EMSA) that consists of twenty Ag/AgCl electrodes distributed in the scalp of the individual. To study the brain response to the stimulus, the animations were shown in two blocks with a five-minute inter-block interval. Each block consisted of 25 biological movement and 25  non-biological movement stimuli presented randomly. Each stimulus was displayed for 1.3 seconds, followed by an inter-stimulus interval (ISI) of 5 seconds. In each trial, a fixation cross appeared in the last second of the ISI. 
A total of 100 point light animations were displayed (2 blocks,  2 conditions [biological and non-biological movement], 25 repetitions).

To construct the brain functional networks, for each subject, phase and repetition of the experiment we first computed a Spearman correlation between each pair of electrodes for each temporal window  $[t, t+333ms]$, for values of $t$ varying every 16.66ms (this corresponds to the \emph{interaction criterion} in Fig.~\ref{fig:estimulo}).
The series of correlations for each pair of electrodes $ij$ (and specific for each subject, phase and repetition) will be denoted by
 $\{\rho_{t}^{ij}\colon t=t_1,\dots,t_n\}$.
 For the construction of the graphs we computed a threshold  for each pair of electrodes $ij$ based on this series of correlations and we put an edge between these electrodes if the  absolute value of the correlation  for a given time $t$ was  above this threshold (this step corresponds to the \emph{network criterion}  in Fig.~\ref{fig:estimulo}). That means to say that for each pair of electrodes we selected a different threshold value, and the selection of this threshold was
done in the following way. 
Let $c$ be a constant, $0<c<1$, and let
$q^{ij}_1$ and $q^{ij}_3$ denote the first and third quartiles of the series of correlations $\{\rho_{t}^{ij}\colon t=t_1,\dots,t_n\}$.
For a given time $t$ define
\begin{equation}
g_{ij}^{t}=\left \{
\begin{array}{cl}
1, & \mbox{if } \rho^{ij}_{t}\geq \mathrm{max}(c,q^{ij}_3) \mbox{ or }  \rho^{ij}_{t} \leq \mathrm{min}(-c,q^{ij}_1)\,;
\\
0, & c.c.
\end{array}
\right.
\end{equation}
In this way, the graph of interactions
for time $t$ will be given
by $g^t=(g^t_{ij})_{1 \leq i < j \leq 20}$.

The rationality of the criterion proposed here is that the graphs constructed in this way select the edges between electrodes that behaves  similarly from a statistical point of view, and this is done by imposing the first and third quartile condition. Each correlation between two electrodes fluctuates in time,  then for a given time $t$ we select the ones that are too small (less than $q^{ij}_1$) or large (greater than $q^{ij}_3$). 
 It is interpreted as follows, a given interaction grows if the two brain regions (principally responsible of the signal) are interacting in an excitatory way feeding back the process, or the interaction can decrease if there exist an inhibitory interaction between them.  Both changes are captured by our criterion. The extra condition greater (or less) to the value $c$ ($-c$) is just for obtaining statistical significant correlations. The value chosen for $c$ in this study is 0.5.

The samples of graphs constructed with our method consist
of 132 graphs for the visible phase and 142 for the occlusion phase of biological movement, for each temporal window (after deletion of spurious repetitions, and considering all subjects).  In the same way, we obtained 132 graphs for the visible phase and 137 for the occlusion phase of the non-biological movement, for each temporal window. 
To perform the tests we selected four non-overlapping windows on each phase, $V_1-V_4$ in the visible phase and $O_1-O_4$ in the occlusion phase (see the top of  Fig.~\ref{fig:estimulo}).

\begin{figure}[t!]
\begin{center}
\subfigure[center][$P$-value of visible vs. occlusion phases.]{
\begin{minipage}{\textwidth}
\centering
\begin{tabular}{|l|c|c|c|c|}
\hline \multirow{2}{*}{Visible vs Occlusion}&\multicolumn{4}{c|}{Windows} \\ \cline{2-5}
& $V_1$ vs. $O_1$ & $V_2$ vs. $O_2$ & $V_3$ vs. $O_3$ & $V_4$ vs. $O_4$ \\ \hline \hline
Biological&\textbf{0.0019}	&0.4294	&0.1984	&0.0278\\\hline\hline
Non-biological &\textbf{0.0016} &	0.8278	&0.1249	&0.6673	\\\hline\hline
\end{tabular}
\vspace*{7mm}
\end{minipage}
\label{tab:pvalor_phase}
}
\subfigure[center][Summary graphs for biological movement.]{\includegraphics[scale=0.6,trim=0 120 0 90]{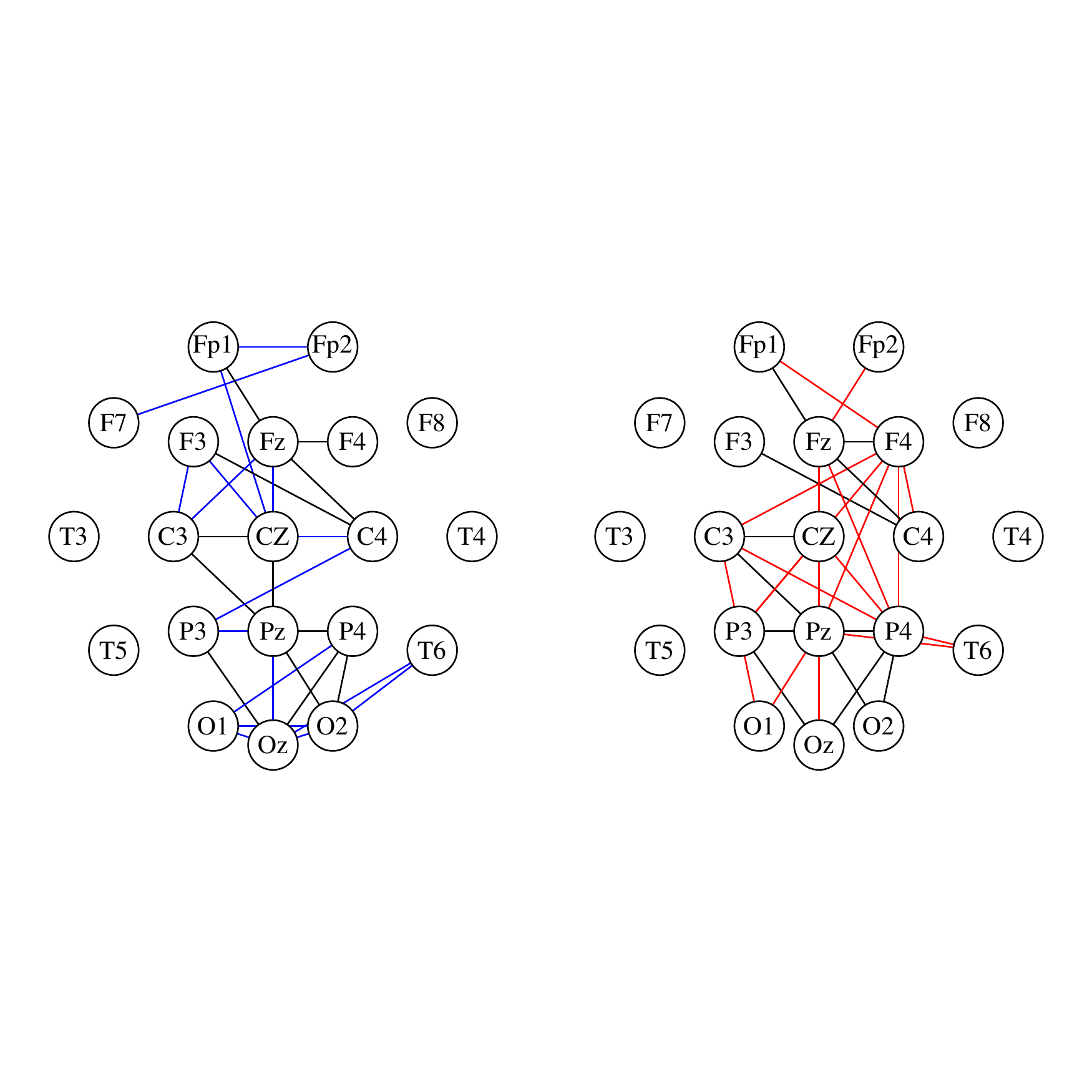}
\label{fig:viXoc_bio}}
\subfigure[center,nooneline][Summary graphs for non-biological movement.]{\includegraphics[scale=0.6,trim=0 120 0 90]{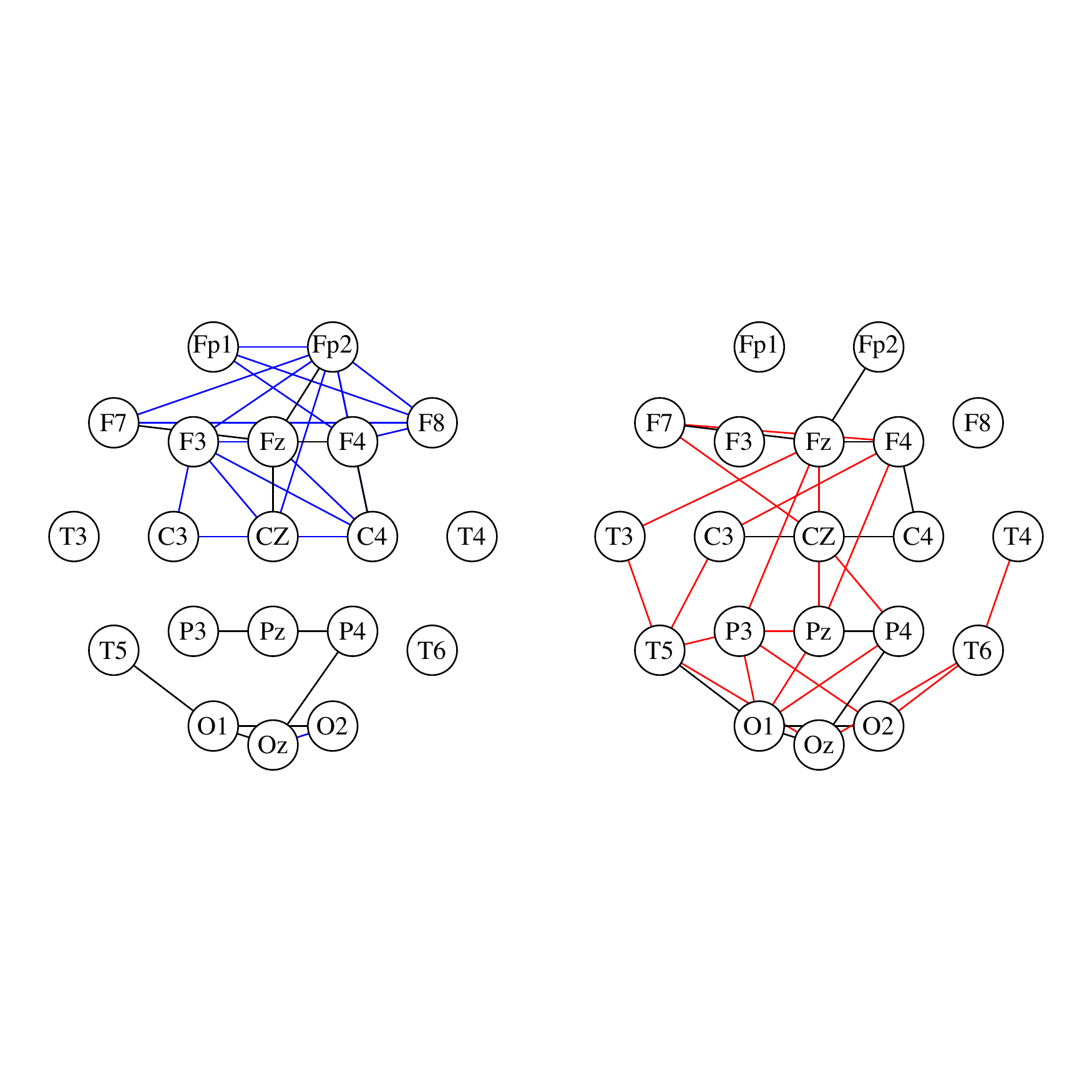}
\label{fig:viXoc_nbio}}
\end{center}
\caption{(a) $P$-value of the test of hypotheses for visible vs. occlusion windows of biological and non-biological movement. (b) Summary graphs of 30 more frequent edges for $V_1$ (left) and $O_1$ (right), for the biological movement. Black edges correspond to common edges. (c) Same as (b) for non-biological movement.}
\label{fig:graph_test.phase}
\end{figure}

We first tested the samples corresponding to visible vs. occlusion windows;
that is we tested $V_1$ vs. $O_1$, $V_2$ vs. $O_2$ and so on, for biological and non-biological movement. To compute the $p$-values, we used a permutation procedure  \cite{manly2007}; that is, for each pair of samples  we 
extract two random subsamples  from the pooled sample, with the same samples sizes of the original ones. Then we compute the test statistic for the subsamples extracted in this way, and we replicate this procedure 1000 times.  The estimated $p$-value  is therefore the empirical proportion of values in the vector of size 1000 built up in this way that are greater than the observed $W$ statistic.  
The $p$-values obtained for the four tests are reported in Fig.~\ref{tab:pvalor_phase}. 
We notice that  in both types of movement the $p$-values corresponding to the first windows of visible and occlusion phases are significantly smaller than the other $p$-values. 
The stimulus onset evokes an event related response   \cite{picton2000} in the first window of the visible phase. This response, also known as visual evoked potential,  is absent in the occlusion phase where there is no stimulus presentation. 
As can be observed the $W$ statistic is able to retrieve  this difference from the graphs distributions.

It is important to remark that the test of hypotheses proposed here does not discriminate which edges in the graphs contribute more significantly to distinguish the two conditions under analysis. Therefore, to compile the results obtained with the test of hypotheses we plotted a summary graph representing each sample by selecting the 30 more frequent edges. Fig.~\ref{fig:graph_test.phase}(b)-(c) illustrate the graphs corresponding to $V_1$ and $O_1$ in the biological and non-biological movement conditions for which the smallest $p$-values were found, as illustrated in Fig.~\ref{tab:pvalor_phase}.

Although the plots of the  30 most frequent edges in the first window of the visible and occlusion phases are quite similar for the biological movement condition, comparatively less edges seem present in the occipital electrodes (O1, Oz and O2) and there is a shift towards the right parietofrontal region in the occlusion period. These results could be taken as an evidence of the  hypothesis raised in \cite{ghislain} that the brain would implicitly ``reenact'' the observed biological movement during the occlusion period (see  for more details).  For the non-biological condition, the 30 most frequent edges in the first window of the visual phase clearly connect electrodes in the frontal region whereas the 30 most frequent edges in the first window of the occlusion phase connect electrodes in the central-occipital region.  

Comparing the biological and non biological conditions during the visible phase, Saunier et al. (2013) \cite{ghislain}  found  differences both  in the right temporo-parietal  and  in centro-frontal regions.  Using functional connectivity, Fraiman et al.  (2014) \cite{fraiman}  confirmed that the left frontal regions may play a major role when it comes to discriminating biological x non biological movements. To confirm these findings we proceeded to test the corresponding windows of the biological and non-biological movement conditions.
 For the visible phase the smallest $p$-value
($<0.03$) was obtained for the third temporal window (time between 668.3ms and 1001.7ms). The occlusion phase does not report significant results in any of the tested windows, see Fig.~\ref{tab:pvalor_nonbio}. We emphasize the fact that this is a more sensible problem compared to the comparison of visible and occlusion windows, in the sense that the differences in the stimuli are very subtle. For that reason it is not surprising that with the actual sample sizes we do not obtain very significant results in this case.

\begin{figure}[t!]
\center
\begin{tabular}{|l|c|c|c|c|}
\hline \multirow{2}{*}{Biological vs. non-biological}&\multicolumn{4}{c|}{Window} \\ \cline{2-5}
& 1 & 2&3&4\\ \hline \hline
Visible &0.1014	&0.6621	&{\bf 0.0295}	&0.5910\\\hline\hline
Occlusion &0.8227	&0.8816	&0.3764	&0.1292	\\\hline\hline
\end{tabular}
\caption{$p$-value of the test of hypotheses for biological vs non-biological movement of visible and occlusion phases.}
\label{tab:pvalor_nonbio}

\end{figure}

\section{Discussion}
In this paper we presented a goodness-of-fit non-parametric test inspired in the recent work \cite{busch2009} for probability distributions over graphs.  To our knowledge this is the first nonparametric goodness-of-fit  test of hypothesis for random graphs distributions.
We derived a closed and efficient formula for the test statistic that implies that the test is consistent against any alternative hypothesis having at least one different marginal distribution over the set of edges. In this case, the simulations show that our test outperforms the simultaneous testing of the marginal means with Bonferroni correction.
As in practice the sample sizes are very small compared with the sample space (in our simulations we took $n=20$ for the sample size versus $2^{28}$ of the sample space for graphs with 8 nodes), our test performs very well  even for small
  differences in the marginal distributions.
  In the real EEG dataset, we showed the potentiality of the $W$ statistic to detect differences in graphs of interaction built from EEG data.

 Although the main focus of this paper is on simple non-directed graphs, the generalization of the test statistic to other graph structures could be possible. This could be done for example  by taking a more general distance function between graphs or by modifying directly the test statistic formula given in Proposition~\ref{main-teo}.  This would generalize the  test to include  other graph structures or would enable the test statistic  to be consistent even for different graph distributions having the same marginals over the edges.

\section{Appendix}

\begin{proof}[Proof of Proposition~\ref{main-teo}]
We will prove the proposition only for the one-sample test. The result for the two-sample test can be derived analogously.
Denote by $w_\g(g)=\bar D_\g (g) -\pi'D(g,\cdot)$. Observe that in order to maximize $|w_\g(g)|$ in   $\G(V)$ it is sufficient to maximize $w_\g(g)$ and $-w_\g(g)$. We have that
 \begin{eqnarray*}
w_\g(g)&=& \dfrac{1}{n}\sum_{k=1}^n D(g,g^{k}) - \sum\limits_{g' \in \G(V)}D(g,g')\pi'(g')\,. \nonumber\\
\end{eqnarray*}
The first sum equals
\begin{eqnarray*}
\frac{1}{n}\sum_{k=1}^n D(g,g^{k})
&=& \frac{1}{n}\sum_{k=1}^n \; \sum_{ij}(g_{ij} - g^k_{ij})^2 \\
&=&  \sum_{ij} ( g_{ij} - 2g_{ij}\overline\g_{ij} + \overline\g_{ij} )\,.
\end{eqnarray*}
The second sum is
\begin{eqnarray*}
 \sum\limits_{g' \in \G(V)}D(g,g')\pi'(g')
 &=& \sum\limits_{g' \in \G(V)}\pi'(g') \sum_{ij}(g_{ij} - g'_{ij})^2  \\
&=&  \sum_{ij}  ( g_{ij} - 2g_{ij}\pi'_{ij} + \pi'_{ij} )\,.
 \end{eqnarray*}
 Therefore we have that
 \begin{eqnarray}\label{swg}
 w_\g(g)&=& \sum_{ij}  (2g_{ij}-1)(\pi'_{ij} -\overline\g_{ij} )\,.
 \end{eqnarray}
As this is a weighted sum, the graph $g^*\in\G(V)$ that maximizes  $w_\g(g)$ is given by
\begin{equation}\label{gstar}
g^*_{ij}=
\begin{cases}
1, & \text{ if }\,\overline{\g}_{ij} \leq \pi'_{ij}
\\
0, & \text{c.c.}
\end{cases}
\end{equation}
Similarly, the graph $g^{**}\in\G(V)$ that maximizes  $-w_\g(g)$ is given by
\begin{equation}\label{gstar2}
g^{**}_{ij}=
\begin{cases}
1, & \text{ if }\,\overline{\g}_{ij} \geq \pi'_{ij}
\\
0, & \text{c.c.}
\end{cases}
\end{equation}
Note also that by a direct calculation from (\ref{swg}) and the definitions (\ref{gstar}) and (\ref{gstar2}) we have that $|w_\g(g)|=|-w_\g(g)|$.
Finally, from (\ref{wg})  and (\ref{gstar}) we obtain
\begin{equation*}
W(\g)\,=\,\max_{g\in\G(V)} |w_\g(g)| \,=\, w_{\g}(g^{*})\,=\,\sum_{ij}
|\overline{\g}_{ij} - \pi'_{ij}|\,.\qedhere
\end{equation*}
\end{proof}

\begin{proof}[Proof of Corollary~\ref{main-cor}]
This is a direct consequence of the multidimensional Central Limit Theorem (cf. Theorem~11.10 in   \cite{breiman1992}).
\end{proof}

\section*{Acknowledgments}
F.L.\/ is partially supported by a CNPq-Brazil fellowship (304836/2012-5) and  FAPESP's fellowship (2014/00947-0). This article was produced as part of the activities of FAPESP Research, Innovation and Dissemination Center for Neuromathematics, grant 2013/07699-0, S\~ao Paulo Research Foundation.
She also thanks L'Or\'eal Foundation for a ``Women in Science'' grant.

\bibliography{./references}

\end{document}